%% file: main.tex
\documentclass[11pt]{article}
\usepackage{style}
\usepackage{shortcuts}
\pdfoutput=1

\author{Amir Abboud\thanks{Weizmann Institute of Science. Part of this work was done while visiting the Courant Institute of Mathematical Sciences at New York University. This work is part of the project CONJEXITY that has received funding from the European Research Council (ERC) under the European Union's Horizon Europe research and innovation programme (grant agreement No.~101078482).}
\and Ron Safier\thanks{Weizmann Institute of Science.}
\and Nathan Wallheimer\footnotemark[2]
}

\date{}

\title{Equivalent Dichotomies for Triangle Detection in Subgraph, Induced, and Colored $H$-Free Graphs}

\begin{document}
    \maketitle
    \begin{abstract}
        A recent paper by the authors (ITCS'26) initiates the study of the Triangle Detection problem in graphs avoiding a fixed pattern $H$ as a subgraph and proposes a \emph{dichotomy hypothesis} characterizing which patterns $H$ make the Triangle Detection problem easier in $H$-free graphs than in general graphs.
        In this work, we demonstrate that this hypothesis is, in fact, equivalent to analogous hypotheses in two broader settings that a priori seem significantly more challenging: \emph{induced} $H$-free graphs and \emph{colored} $H$-free graphs.

        Our main contribution is a reduction from the induced $H$-free case to the non-induced $\Hplus$-free case, where $\Hplus$ preserves the structural properties of $H$ that are relevant for the dichotomy, namely $3$-colorability and triangle count. A similar reduction is given for the colored case.

        A key technical ingredient is a self-reduction to Unique Triangle Detection that preserves the induced $H$-freeness property, via a new color-coding-like reduction.
    \end{abstract}

    \clearpage
    \thispagestyle{empty} 
    \tableofcontents
    \clearpage
    \setcounter{page}{1}

    \section{Introduction}
        \label{sec:introduction}
        \input{sections/introduction.tex}

    \section{Related Work}
        \label{sec:related}
        \input{sections/related.tex}

    \section{Preliminaries}
        \label{sec:preliminaries}
        \input{sections/preliminaries.tex}

    \section{\texorpdfstring{Induced $H$-Free to $\Hplus$-Free}{Induced H-Free to H+-Free}}
        \label{sec:induced}
        \input{sections/induced.tex}

    \section{\texorpdfstring{Colored $H$-Free to $\Hstar$-Free}{Colored H-Free to H*-Free}}
        \label{sec:colored}
        \input{sections/colored.tex}

    \section{Equivalence of the Three Hypotheses}
        \label{sec:equiv}
        \input{sections/equiv.tex}

    \section{Applications for Specific Patterns}
        \label{sec:applications}
        \input{sections/applications.tex}

    \section{\texorpdfstring{Generalization to $k$-Clique Detection}{Generalization to k-Clique Detection}}
        \label{sec:k-clique}
        \input{sections/k_clique.tex}

    \section{Conclusions}
        \label{sec:conclusions}
        \input{sections/conclusions.tex}

    \section*{Acknowledgments}
        We thank Shyan Akmal for inspiring discussions.

    \bibliographystyle{alpha}
    \bibliography{references}

\end{document}

%% file: sections/introduction.tex
The Triangle Detection problem asks to decide if a given $n$-node graph $G$ contains a triangle. It is arguably the
simplest problem without a near-linear time algorithm (i.e., running in $m^{1+o(1)}$ time in connected graphs with $m$
edges) and serves as the underlying ``complexity atom'' at the core of fine-grained hardness assumptions. For dense
inputs (e.g. where $m = \Omega(n^2)$), the trivial algorithm runs in $\Order(n^3)$ time. A truly subcubic
$n^{3-\Omega(1)}$ running time~\cite{itai1977finding} can be achieved by algebraic techniques: $\Order(n^{\omega})$
where $\omega < 2.371339$ via Fast Matrix Multiplication (FMM)~\cite{alman2025more} or a higher bound via Fast Fourier
Transform (FFT)~\cite{uffenheimer2025approximate,CohnKSU05}. The best bound achievable by other, so-called
``combinatorial'' methods~\cite{arlazarov1970economical,BW12,chan2014speeding,yu2018improved,abboud2024new} is mildly
subcubic $n^3/2^{\log^{\Omega(1)}n}$. Faster, classical algorithms are known for sparse
graphs~\cite{itai1977finding,alon1997finding}, achieving time $\Order(m^{3/2})$ combinatorially and
$\Order(m^{\frac{2\omega}{\omega+1}})$ through FMM.

\medskip

Recently, the authors~\cite{abboud2026triangle} posed the following natural question:
\paragraph*{Main Question:}{\emph{For which patterns $H$ can we solve Triangle Detection in $H$-free graphs faster than
in general graphs?}}

\medskip

While this question can be asked with respect to each of the aforementioned classical bounds for general
graphs,~\cite{abboud2026triangle} suggests starting with the clean question of whether the $H$-freeness restriction
allows one to break the cubic barrier with a combinatorial algorithm. Notably, due to a reduction by Vassilevska and
Williams~\cite{williams2010subcubic}, breaking this barrier in general graphs is equivalent to refuting the famous
Boolean Matrix Multiplication (BMM) conjecture.

\begin{conjecture}[BMM]
   \label{conj:bmm}
   There is no truly subcubic combinatorial Triangle Detection (equivalently, BMM) algorithm (in general $n$-node
   graphs).
\end{conjecture}

Thus, the main question above can be restated as follows:

\begin{openquestion}
\label{oq1}
For which patterns $H$ can Triangle Detection in $H$-free graphs be solved combinatorially in truly subcubic time?
\end{openquestion}

A major motivation for studying combinatorial algorithms for Triangle Detection generally, and Open Question~\ref{oq1}
specifically, is the following. The algebraic $n^{3-\Omega(1)}$ algorithms for Triangle Detection (via FMM or FFT) are
well-known to fail to generalize for important variants of the basic problem, and one may hope that a different,
``combinatorial'' algorithm would not suffer from the same limitations. Such an achievement would have groundbreaking
consequences, even if the new algorithm only overcomes \emph{some} of these limitations. For example, a truly subcubic
algorithm that generalizes to the \emph{minimum weight} setting would refute the All-Pairs Shortest-Paths (APSP)
Conjecture, an algorithm that generalizes to the \emph{listing} setting would refute the $3$-SUM Conjecture, an
algorithm for the \emph{online} setting would refute the Online Matrix Vector Multiplication (OMv) hypothesis, and so
on. For this reason, we follow the convention in the literature on combinatorial algorithms of leaving the notion
flexible: as explained in~\cite{abboud2024new}, the goal is to find new techniques that would generalize to \emph{any}
setting of interest. See~\cite[Section 1.1]{abboud2024new} for a broader discussion and additional motivations for
studying combinatorial algorithms. Thus, the quest to develop a truly subcubic combinatorial algorithm for Triangle
Detection is at the forefront of algorithmic attacks on the conjectures of fine-grained complexity.\footnote{While
fine-grained complexity is often presented as showing conditional lower bounds for a vast number of problems assuming
the hardness of a small set of core problems, one should also keep in mind the interpretation that algorithmic research
on a vast number of problems should be re-directed to the few core problems, because progress will be made only by
finding new techniques that break the conjectures for the core problems.}\footnote{Note that a ``combinatorial''
algorithm is not guaranteed to achieve any of these consequences, but one hopes that it would achieve at least one of
them.} A promising line of attack, related to modern algorithmic paradigms such as expander
decompositions~\cite{saranurak2019expander, abboud2023worst,abboud2024worst} and regularity
lemmas~\cite{BW12,abboud2024new}, is via a ``structure vs. randomness'' approach where a general graph is decomposed
into random parts -- for which the problem becomes easy -- and structured parts (and a sparse part). The question
becomes: which notion of structure can be used in such an approach? (See~\cite{abboud2026C4Detection} for a novel notion
recently employed successfully in the context of subgraph detection.) One natural candidate is $H$-freeness, and Open
Question~\ref{oq1} aims to identify the cases in which it is useful.

In the same paper~\cite{abboud2026triangle}, the authors present both positive and negative results, leading them to
suggest the following ``dichotomy hypothesis'' as a crisp characterization of the answer to Open Question~\ref{oq1}.

\begin{hypothesis}[The Dichotomy Hypothesis for Triangle Detection in $H$-Free Graphs]
   \label{hypothesis:subgraph}
   Triangle Detection in $H$-free graphs is:
   \begin{itemize}
      \item ``easy'', i.e., solvable in $n^{3-\Omega(1)}$ time via a combinatorial algorithm, if $H$ is
         \emph{$3$-colorable and contains at most one triangle}; and
      \item ``hard'', i.e., requires $n^{3-o(1)}$ time combinatorially, otherwise.
         \footnote{In~\cite{abboud2026triangle}, the authors define ``hard'' to mean that it is as hard as the general
         case, meaning that there is a BMM-based lower bound. This makes it technically possible that a pattern $H$ is
         neither ``easy'' nor ``hard'' (e.g., if one can only show a lower bound under a different assumption). For
         simplicity, we choose the more relaxed statement.}
   \end{itemize}
\end{hypothesis}

The results of~\cite{abboud2026triangle} confirm this hypothesis for any pattern of size up to $8$~\cite[Proposition
1.3]{abboud2026triangle}. Their hardness reductions apply if $H$ is not $3$-colorable or if it contains more than one
triangle, i.e., to all the ``hard'' cases in the hypothesis. However, their algorithmic results can only handle a subset
of the (supposedly) ``easy'' cases. In particular, a $3$-colorable pattern $H$ with one triangle is proven to be
``easy'' if it admits a \emph{degenerate coloring} (see Definition~\ref{def:degenerate}).

\subsection{\texorpdfstring{What about \emph{induced} $H$-free graphs?}{What about induced H-free graphs?}}

   While~\cite{abboud2026triangle} start the investigation into Open Question~\ref{oq1} by focusing on the (simpler)
   non-induced \emph{subgraph} setting of $H$-freeness, the investigation is equally natural and well-motivated for the
   (typically more complex) \emph{induced} setting. Note that the class of induced $H$-free graphs generalizes $H$-free
   graphs: every $H$-free graph is also an induced $H$-free graph, but the converse does not hold. For instance, a
   clique of size $n$ is induced $P_k$-free, where $P_k$ is the $k$-vertex path, but it contains $P_k$ as a subgraph
   (assuming $3 \leq k \leq n$).

   The main question we ask is: \emph{how does the answer to Open Question~\ref{oq1} change in the induced setting?}
   Because $H$-free graphs are a subset of induced $H$-free graphs, any lower bound for Triangle Detection in $H$-free
   graphs also extends to the induced setting, i.e., to \emph{Triangle Detection in induced $H$-free graphs}. However,
   the same cannot be said about the upper bounds. In particular, the algorithmic techniques that were used
   in~\cite{abboud2026triangle} completely fail for the induced case, as we explain next.

   The primary algorithmic technique employed in~\cite{abboud2026triangle} is the \emph{embedding approach}. This method
   proceeds by iteratively embedding vertices of the pattern $H$ into the host graph $G$, which is assumed to be
   $3$-colored by color-coding~\cite{alon1995color}. A critical step in this recursion involves restricting the search
   space to the neighborhood of the currently embedded vertex. Specifically, when mapping a vertex $v \in V(H)$ to a
   vertex $u \in V(G)$, the algorithm restricts a specific color class of $G$ to the neighbors of $u$, effectively
   making $u$ universal to the remaining vertices in that color class. The paper shows that this reduces the density of
   $H$-free graphs for patterns that admit a degenerate coloring.

   However, the embedding approach is fundamentally incompatible with the induced setting. In the induced context, if
   the pattern $H$ contains a non-edge between $v$ and some unembedded vertex $w$, we are forbidden from mapping $w$ to
   a neighbor of $u$. Consequently, this technique is limited to a very restricted family of induced patterns.

   As a result, even for simple patterns like paths, results for the induced setting remain elusive. For general $k$, a
   combinatorial $n^{3-\Omega(1)}$ algorithm for Triangle Detection in induced $P_k$-free graphs is unknown. This stands
   in stark contrast to the subgraph setting, where (non-induced) $P_k$-free graphs have bounded arboricity by the
   Erd\H{o}s-Gallai theorem~\cite{gallai1959maximal}, allowing Triangle Detection to be solved trivially in $\Order(m) =
   \Order(n)$ time using the classic Chiba-Nishizeki algorithm~\cite{chiba1985arboricity}. This leads one to suspect
   that a classification that would answer Open Question~\ref{oq1} in the induced case would be very different
   from~\cref{hypothesis:subgraph}. At the very least, such a dichotomy appears significantly more challenging to
   establish.

   \paragraph*{Main Result:} The main result of this work is a reduction from the induced case to the non-induced case,
   showing a surprising \emph{equivalence} between the dichotomy hypotheses in the two settings.

   \begin{theorem}[Informal, see Theorem~\ref{thm:main}]
      \label{informal:induced-to-subgraph}
      The Dichotomy Hypothesis for Triangle Detection in $H$-free graphs holds \emph{if and only if} the same dichotomy
      hypothesis holds for induced $H$-free graphs.
   \end{theorem}

   More explicitly, we show a reduction from Triangle Detection in induced $H$-free graphs to (non-induced)
   $\Hplus$-free graphs, such that if $H$ falls into the ``easy'' cases of Hypothesis~\ref{hypothesis:subgraph}, then so
   does $\Hplus$. Consequently, the same dichotomy hypothesis of~\cite{abboud2026triangle} not only answers the
   non-induced case but also the induced case. In other words,~\cref{hypothesis:subgraph} turns out to be tantamount to
   the following hypothesis (in which we only added the word ``induced'').

   \begin{hypothesis}[The Dichotomy in the Induced Setting]
      \label{hypothesis:induced}
      Triangle Detection in \emph{induced} $H$-free graphs is:
      \begin{itemize}
         \item ``easy'', i.e., solvable in $n^{3-\Omega(1)}$ time via a combinatorial algorithm, if $H$ is
            \emph{$3$-colorable and contains at most one triangle}; and
         \item ``hard'', i.e., requires $n^{3-o(1)}$ time combinatorially, otherwise.
      \end{itemize}
   \end{hypothesis}

   The main message of this result is that the induced setting will not be any more challenging than the non-induced
   case. From a positive perspective, this means that all we have to do to answer Open Question~\ref{oq1} for both
   non-induced \emph{and} induced is to find algorithms for the remaining ``easy'' cases of the (non-induced)
   hypothesis. From a negative perspective, to refute the dichotomy hypothesis proposed by~\cite{abboud2026triangle}, it
   is enough to find a hardness result for an \emph{induced} pattern.

   Our reduction can also be used to prove new positive results for specific patterns $H$ that map to cases already
   solvable by the algorithms of~\cite{abboud2026triangle} (i.e., those with a degenerate coloring). Interesting
   applications include the induced $P_6$ and induced $C_7$; see Section~\ref{sec:applications}.

   A main technical ingredient in our proof, which may be of independent interest, is a new ``color-coding'' style
   self-reduction for Triangle Detection that preserves the induced $H$-freeness.

   Finally, we note that our reduction also applies to the generalization of the problem to $k$-Clique Detection for any
   fixed $k \geq 3$, showing that the generalized dichotomy hypotheses for $k$-Clique Detection in $H$-free graphs and
   in induced $H$-free graphs are also equivalent. We discuss this generalization in more detail in \cref{sec:k-clique}.

\subsection{\texorpdfstring{Colored $H$-free graphs}{Colored H-free graphs}}

   Next, we consider another natural variant that appears incomparable to either the non-induced setting or the induced
   setting. The notion of \emph{colored $H$-freeness} was introduced by~\cite{abboud2026triangle} to exploit the
   structural properties of patterns that admit a degenerate coloring. Roughly speaking, a colored graph $G$ is
   considered free of a colored pattern $H$ if it contains no copy of $H$ that preserves its colors. The formal
   definition for $3$-colored graphs is provided below, though the concept generalizes naturally to general colorings.
   \begin{definition}[Colored copy of $H$ in $G$]
      \label{def:colored-copies}
      Given a $3$-colored graph $G$ with coloring $c_G \colon V(G) \to \Colors$, and a $3$-colored pattern $H$ with
      coloring $c_H \colon V(H) \to \Colors$, a \emph{colored copy of $H$} is an injective homomorphism $\phi \colon
      V(H) \to V(G)$ that preserves colors, i.e., $c_G(\phi(v)) = c_H(v)$ for every $v \in V(H)$. If $G$ contains no
      such copy, we say that $G$ is \emph{colored $H$-free}.
   \end{definition}

   Much like induced $H$-freeness, colored $H$-freeness is a generalization of the standard subgraph setting. It is
   important to note that a colored graph $G$ may be colored $H$-free even if it contains many copies of $H$ in the
   uncolored sense, provided the coloring of $G$ does not induce the specific coloring of $H$ on those copies. For
   instance, consider the specific $3$-coloring of a $6$-cycle (denoted $C_6$) depicted in \cref{fig:C6}.

   \begin{figure}[htbp]
      \centering
      \includegraphics[scale=1.0]{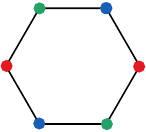}
      \caption{A specific $3$-coloring of $C_6$.}
      \label{fig:C6}
   \end{figure}

   It is a classical result that graphs excluding an uncolored $C_6$ as a subgraph must be
   sparse~\cite{bondy1974cycles}, allowing for an efficient Triangle Detection algorithm. In stark contrast, graphs
   avoiding a \emph{colored} $C_6$ can be dense. For example, consider the colored blowup of a $C_9$ in
   \cref{fig:C9-blowup}.

   \begin{figure}[htbp]
      \centering
      \includegraphics[scale=1.0]{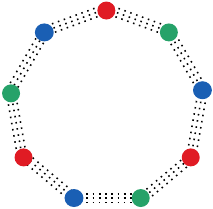}
      \caption{A colored blowup of $C_9$. This graph is obtained by replacing each vertex of $C_9$ with an independent
      set of size $n/9$ and every edge with a complete bipartite graph. All vertices within a single blown-up set share
      the same color.}
      \label{fig:C9-blowup}
   \end{figure}

   This graph is free of the colored $C_6$ pattern shown in \cref{fig:C6}, yet it contains $\Theta(n^2)$ edges.
   Consequently, it contains many copies of $C_6$ whose colorings fail to respect the color constraints. We remark that
   the embedding framework of~\cite{abboud2026triangle} fails to yield a truly subcubic algorithm for the colored $C_6$
   in \cref{fig:C6} because this coloring is not degenerate.

   Furthermore, observe that the construction in \cref{fig:C9-blowup} exhibits high ``colored density'': every vertex
   has $\Omega(n)$ neighbors in every color class other than its own. This property is significant because the
   algorithms in~\cite{abboud2026triangle} rely on efficiently eliminating vertices with low colored degree. For a dense
   graph like the colored $C_9$-blowup, where no such vertices exist, it is unclear how to detect triangles efficiently.

   This suggests that the colored setting may pose strictly harder challenges than the uncolored setting, making an
   answer to Open Question~\ref{oq1} for the colored setting seem further away.

   Our second result is another reduction showing that the colored case is \emph{also equivalent} to the non-induced and
   induced cases, in the sense that their dichotomy hypotheses stand or fall together.

   \begin{theorem}[Informal, see Theorem~\ref{thm:main}]
      The Dichotomy Hypothesis for Triangle Detection in $H$-free graphs holds \emph{if and only if} the same dichotomy
      hypothesis holds for \emph{colored} $H$-free graphs.
   \end{theorem}
   Specifically, the Triangle Detection problem in colored $H$-free graphs is defined as follows: the input is a colored
   $H$-free graph $G$ (in particular, the input $G$ is provided together with a $3$-coloring), and the task is to decide
   whether $G$ contains a triangle.

   Just like in the induced case, we show a reduction from Triangle Detection in colored $H$-free graphs to (uncolored)
   $\Hstar$-free graphs, such that if $H$ falls into the ``easy'' cases of~\cref{hypothesis:subgraph} (regardless of its
   coloring), then so does $\Hstar$. Consequently, the same characterization suggested by~\cite{abboud2026triangle}
   extends to the colored case.

   \begin{hypothesis}[The Dichotomy in the Colored Setting]
      \label{hypothesis:colored}
      Triangle Detection in \emph{colored} $H$-free graphs is:
      \begin{itemize}
         \item ``easy'', i.e., solvable in $n^{3-\Omega(1)}$ time via a combinatorial algorithm, if $H$
            \emph{contains at most one triangle}; and
         \item ``hard'', i.e., requires $n^{3-o(1)}$ time combinatorially, otherwise.
      \end{itemize}
   \end{hypothesis}
   Unlike in the subgraph and induced settings, the colored dichotomy has no separate $3$-colorability condition in the
   ``easy'' case: the pattern $H$ is already $3$-colored by~\cref{def:colored-copies}. Moreover, the criterion depends
   only on the underlying uncolored graph $H$, and not on which $3$-coloring of $H$ is prescribed.

   \medskip
   To conclude, our results for the colored and induced settings imply our main result:
   \begin{theorem}
      \label{thm:main}
      The three hypotheses: \labelcref{hypothesis:subgraph}, \labelcref{hypothesis:induced,hypothesis:colored}, are all
      equivalent.
   \end{theorem}

\subsection{Technical overview}
   We begin by presenting the high-level idea behind our result for the induced setting. Given a pattern $H$, we begin
   with the following observation: If an induced $H$-free graph $G$ contains $H$ as a subgraph, the embedded copy of $H$
   in $G$ must also induce a non-edge of $H$ (otherwise, it would be an induced copy of $H$). We exploit this
   observation by augmenting $H$ with additional structure that, together with a non-edge of $H$ in $G$, implies the
   existence of many triangles in $G$. This makes the augmented pattern unlikely to exist in $G$: Intuitively, the
   ``hardest'' instances of detection problems like Triangle Detection have at most one witness.

   However, proving this intuition formally in the setting of induced $H$-free graphs presents technical difficulties.
   Specifically, what we need is a self-reduction to the \emph{unique} variant of the problem that preserves the induced
   $H$-freeness property. In the \emph{Unique Triangle Detection} problem, the input graph is promised to contain at
   most one triangle.

   The standard approach for proving self-reductions to the unique variant of Triangle Detection (and related problems
   like $k$-Clique Detection) relies on the input being color-coded. Standard color coding for Triangle Detection
   (e.g.,~\cite{alon1995color}) is usually done either by having $3$ independent copies of $V(G)$ and $3$ copies of
   $E(G)$ appearing between them, or by randomly partitioning the vertices into three buckets and deleting edges within
   each bucket. However, such transformations do not necessarily preserve the induced $H$-freeness property of the
   graph. We provide an alternative color-coding approach that preserves the $H$-freeness property. Specifically, we
   show a self-reduction from Triangle Detection in induced $H$-free graphs, to Triangle Detection in \emph{$3$-colored}
   induced $H$-free graphs:

   \begin{restatable}[\texorpdfstring{Color-coding induced $H$-free graphs}{Color-coding induced H-free
      graphs}]{theorem}{colorcoding}
      \label{lem:color-coding}
      There is a deterministic algorithm that reduces Triangle Detection to $3$-colored Triangle Detection, running in
      $\Order(n^{2.5})$ time. It either finds a triangle in $G$, or outputs $\Order(n^{3/2})$ instances, where each
      instance is a $3$-colored induced subgraph of $G$ on $\Order(\sqrt{n})$ vertices. Every triangle in $G$ is present
      in exactly one instance.
   \end{restatable}

   Note that, unlike standard color-coding, the output of this algorithm consists of \emph{induced} subgraphs of $G$
   (i.e., no edge deletions). Hence, if $G$ is induced $H$-free, then every instance in the output is induced $H$-free.
   Having this tool in hand, we can then apply a self-reduction to Unique Triangle Detection, such that every output
   instance is free from the augmented pattern of $H$.

   \paragraph*{Colored $H$-free graphs.} In the setting of colored $H$-free graphs, our strategy is to embed uncolored
   copies of $H$ within a larger $3$-colorable pattern $\Hstar$, but augment it to ensure that \emph{every} $3$-coloring
   of $\Hstar$ forces the specific forbidden coloring of $H$ embedded within $\Hstar$. Thus, if a $3$-colored graph $G$
   contains $\Hstar$ as a subgraph, the coloring that $G$ induces on $\Hstar$ will imply the existence of a forbidden
   colored copy of $H$ in $G$. This idea is implemented using various gadgets to enforce that different copies of $H$ in
   $\Hstar$ receive different colorings.

%% file: sections/related.tex
Various notions of $H$-freeness have been widely studied to understand the complexity of fundamental problems across
different computational regimes. The three primary notions prevalent in the literature are $H$-minor-freeness,
$H$-freeness (the subgraph setting), and induced $H$-freeness. Observe that the classes of graphs avoiding these
patterns form a strict hierarchy: $H$-minor-free $\subset$ $H$-free $\subset$ induced $H$-free.

Because $H$-minor-free graphs constitute a highly restricted family, many fundamental problems that are presumed hard on
general graphs become significantly easier in $H$-minor-free graphs for every fixed $H$ (yielding trivial dichotomies
where all patterns are classified as ``easy''). Examples in the coarse complexity regime include Max-Cut and Subgraph
Isomorphism, which become solvable in polynomial time~\cite{kaminski2012max, grohe2012fixed,
ponomarenko1988isomorphism}. In the fine-grained context, there are fast algorithms for Diameter, Distance Oracles, and
Decremental Reachability~\cite{ducoffe2022diameter, duraj2024better, le2024vc, karczmarz2025subquadratic}. For Triangle
Detection (and, more generally, $k$-Clique Detection), $H$-minor-free graphs have bounded
degeneracy~\cite{thomason1984extremal, kostochka1982minimum}. Consequently, the problem can be solved in linear time for
every pattern $H$ using the Chiba-Nishizeki algorithm, yielding a trivial dichotomy.

In the intermediate setting of $H$-free graphs, dichotomies are known for Max-Cut, Maximum Independent Set,
List-Coloring, and (very recently) Steiner
Forest~\cite{kaminski2012max,alekseev1992complexity,golovach2012coloring,eaglingvose2026steiner}. In the fine-grained
context, problems such as Diameter and APSP in dense graphs admit trivial dichotomies between cubic and truly subcubic
time. This arises because hard instances for these problems include bipartite graphs via well-known
reductions~\cite{roditty2013fast, williams2010subcubic}. Therefore, these problems remain hard for any non-bipartite
pattern $H$, whereas for bipartite patterns, the problems become easy due to the sparsity of $H$-free
graphs~\cite{kovari1954problem}. However, other dichotomies between linear and non-linear time can be established.
Recently, Johnson et al.~\cite{johnson2025complexity} established a dichotomy for Diameter in $H$-free graphs between
patterns solvable in linear time and those requiring quadratic time. We remark here that in the context of Triangle
Detection, a dichotomy between linear and non-linear time is conditionally established: Abboud et
al.~\cite{abboud2022hardness} show a $m^{1+\Omega(1)}$ lower bound for Triangle Detection in $C_4$-free graphs, and they
pose as an open question an extension to $C_k$-free graphs for every fixed $k \geq 4$. Assuming this holds, the ``easy''
cases (i.e., linear) are all forests, since such graphs have bounded arboricity by the Erd\H{o}s-Gallai theorem (so the
Chiba-Nishizeki algorithm solves them in linear time), and the ``hard'' cases (i.e., non-linear) are all patterns that
contain a cycle (i.e., non-forests).

In the most general setting of induced $H$-free graphs within the coarse complexity regime, dichotomies are known for
Coloring, Dominating Set, and Max-Cut~\cite{kral2001complexity, korobitsin1992complexity, kaminski2012max}. A long line
of work exists regarding the Maximum Independent Set problem in induced $H$-free graphs~\cite{alekseev1992complexity,
abrishami2022polynomial, gartland2024maximum, majewski2024max}, but a P/NP dichotomy has yet to be fully established.
The problem was also considered in approximation and parameterized settings~\cite{chudnovsky2020quasi,
bonnet2020parameterized}. In the fine-grained regime, the work of Oostveen et al.~\cite{oostveen2025complexity} implies
a dichotomy for Diameter between linear-time and non-linear-time patterns, assuming the optimality of current algorithms
for the Orthogonal Vectors problem in the high-dimensional setting (i.e., where $d = \Theta(n)$). They also demonstrate
an $m^{2\omega/(\omega+1)-o(1)}$ (or combinatorial $m^{3/2-o(1)}$) lower bound that holds for almost every ``hard''
pattern, conditioned on the conjecture that finding a simplicial vertex requires this amount of time. Note that in
weighted graphs, problems such as Diameter and APSP admit trivial dichotomies where all patterns are classified as
``hard''. This is because edges with infinite weight can be added to effectively form a clique (additionally, since
these problems admit hard bipartite instances, they remain hard even when excluding small cliques). Weighted Max-Cut
also admits a trivial dichotomy for the same reason; furthermore, it is known that the problem remains NP-hard even on
triangle-free graphs.\footnote{That is because a $2$-subdivision preserves the maximum cut up to an additive term of
$2m$.} For $k$-Clique Detection, the problem is known to be hard even for very small forbidden patterns. This follows
from results on the dual $k$-Independent Set problem~\cite{bonnet2020parameterized} (obtained via graph
complementation), although several positive algorithmic results have been achieved for specific
cases~\cite{husic2019independent, bonnet2020parameterized}.

%% file: sections/preliminaries.tex
Throughout, let $G=(V(G),E(G))$ be a graph with $n=|V(G)|$ vertices and $m = |E(G)|$ edges. An edge between vertices $u$
and $v$ is denoted by $uv$, and a triangle on vertices $x,y,z$ is denoted by $xyz$. If $uv \notin E(G)$, then we say
that $uv$ is a \emph{non-edge} of $G$. A graph is said to be \emph{properly $3$-colored} if each vertex is assigned one
of the three colors from the set $\Colors$, such that no two adjacent vertices share the same color. This extends to
$k$-colorings as well, where the set of colors is simply $\set{1,2,\ldots,k}$. Throughout the paper, we often omit the
word ``properly''; thus, a $k$-colored graph means a properly $k$-colored graph unless explicitly stated otherwise. A
\emph{color class} of a $k$-coloring is the set of vertices assigned a particular color. For a vertex $v \in V(G)$, its
neighborhood is denoted by $N(v)$.

A \emph{copy of $H$ in $G$} is an injective homomorphism $\phi \colon V(H) \to V(G)$; that is, $\phi$ is one-to-one and
for every edge $xy \in E(H)$, we have $\phi(x)\phi(y) \in E(G)$. We say that $G$ is \emph{$H$-free} if it contains no
copies of $H$. An \emph{induced copy of $H$ in $G$} is a mapping $\phi \colon V(H) \to V(G)$ such that $\phi$ is
one-to-one and $xy \in E(H)$ if and only if $\phi(x)\phi(y) \in E(G)$. We say that $G$ is \emph{induced $H$-free} if it
contains no induced copies of $H$.

Throughout this paper, we assume the standard Word-RAM model of computation with words of size $\Order(\log n)$. We use
$\tOrder(f(n))$ to hide polylogarithmic factors; that is, $\tOrder(f(n)) = \Order(f(n) \log^c n)$ for some constant $c$.

%% file: sections/induced.tex
In this section, we prove the following theorem:
\begin{restatable}{theorem}{inducedtosubgraph}
   \label{thm:induced-to-subgraph}
   For every $3$-colorable pattern $H$ containing at most one triangle, there exists a $3$-colorable pattern $\Hplus$
   containing the same number of triangles as $H$, such that Triangle Detection in induced $H$-free graphs reduces to
   Triangle Detection in $\Hplus$-free graphs. Specifically, suppose there is an algorithm for Triangle Detection in
   $\Hplus$-free graphs that runs in time $T(N)$ on $N$-vertex graphs and succeeds with probability at least $2/3$. In
   that case, there is an algorithm for Triangle Detection in induced $H$-free graphs running in time \[
   \tOrder\left(n^{2.5} + n^{3/2} \cdot T(\sqrt{n}) \right), \] and succeeding with probability at least
   $1-\Order(1/n)$.
\end{restatable}
\noindent Note that the $n^{3/2} \cdot T(\sqrt{n})$ term in the time complexity dominates the $n^{2.5}$ term whenever
$T(N) = \tilde\Omega(N^2)$ (which corresponds to linear time in the input size for dense graphs).

The definition of the augmented pattern $\Hplus$ is given below.
\begin{definition}
   Given a pattern $H$, the \emph{augmented pattern} $\Hplus$ is obtained by attaching two wedges (paths of length $2$)
   to the endpoints of every non-edge of $H$. Formally, $V(H) \subseteq V(\Hplus)$, $E(H) \subseteq E(\Hplus)$, and for
   every $u,v \in V(H)$ such that $uv \notin E(H)$, we add two distinct vertices $x_{uv}, y_{uv}$ to $V(\Hplus)$ and
   four edges $u x_{uv}, v x_{uv}, u y_{uv}, v y_{uv}$ to $E(\Hplus)$.
\end{definition}
Crucially, we observe that if $H$ is $3$-colorable and contains at most one triangle, then $\Hplus$ also satisfies these
properties. That is because the additional vertices $x_{uv}$ and $y_{uv}$ do not form a triangle by definition.
Moreover, since their only neighbors are two vertices in $H$, they can be assigned a color distinct from those of their
neighbors. Another crucial property of $\Hplus$ is that it contains $H$ as a subgraph. Consequently, if an induced
$H$-free graph contains $\Hplus$ as a subgraph, the copy of $H$ embedded within must contain at least one additional
edge corresponding to a non-edge of $H$. By construction, this non-edge closes two triangles in $\Hplus$. See
\cref{fig:P_5-F} for an example.

\begin{figure}[htbp]
   \centering
   \includegraphics[scale=1.2]{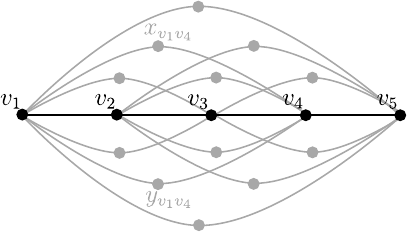}
   \caption{The augmented pattern $\Hplus$ corresponding to $P_5$. The original pattern $P_5$ is drawn in black, and the
   additional parts are drawn in gray. Observe that if an induced $P_5$-free graph contains $\Hplus$ as a subgraph, then
   the vertices corresponding to $P_5$ must be connected by at least one non-edge of $P_5$, e.g., the edge $v_1v_4$. In
   that case, the graph contains at least two triangles: $v_1v_4x_{v_1v_4}$ and $v_1v_4y_{v_1v_4}$.}
   \label{fig:P_5-F}
\end{figure}

The goal is, therefore, to show a self-reduction to the unique variant of Triangle Detection that preserves the induced
$H$-freeness property. Specifically, we prove the following lemma:
\begin{lemma}[Reduction to Unique Triangle Detection]
   \label{lem:reduction-to-unique}
   There is a randomized reduction that, given a graph $G$, either finds a triangle in $G$, or outputs $k =
   \tOrder(n^{3/2})$ induced subgraphs $G_1, \ldots, G_k$, such that:
   \begin{enumerate}
      \item If $G$ contains a triangle, then with probability $1 - \Order(1/n)$, there exists at least one $G_i$ that
         contains exactly one triangle.
      \item Each $G_i$ has $\Order(\sqrt{n})$ vertices.
   \end{enumerate}
   The reduction runs in $\tOrder(n^{2.5})$ time.
\end{lemma}
Note that since each $G_i$ is an induced subgraph of $G$, the property of induced $H$-freeness is preserved. Namely, if
$G$ is induced $H$-free, then so is $G_i$ for every $1 \leq i \leq k$. To prove this lemma, we use color coding followed
by an isolation step. We begin with our proof of the color-coding result: \colorcoding*
\begin{proof}[Proof of \cref{lem:color-coding}]
   The idea is to find a coloring of $G$ with $\Order( \sqrt{n} )$ colors, and then take every triplet of color classes
   as an instance. A paper by Adrian Dumitrescu~\cite{dumitrescu2025finding} presents an algorithm for finding either a
   triangle in $G$ or a $\Order(\sqrt{n})$-coloring in $\Order(m + n^{3/2}) = \Order(n^2)$ time~\cite[Theorem
   4]{dumitrescu2025finding}.~\footnote{The running time can be improved to $\Order(m+n)$ by a more careful
   implementation; see the $k=3$ case of \cref{lem:clique-or-coloring} for the details in a more general form.} If the
   algorithm does not find a triangle, we refine the provided coloring as follows. As long as some color class is larger
   than $2 \sqrt{n}$, we extract a block of size $\sqrt{n}$, and give it a new color. This results in a finer partition
   where each color class has size $\Order( \sqrt{n})$. Note that the refinement procedure adds at most $\Order(
   \sqrt{n} )$ additional colors because each block has size $\sqrt{n}$ and the blocks are disjoint.

   Now, we create $3$-colored Triangle Detection instances by taking the induced subgraph of every triplet of color
   classes. There are $\Order(n^{3/2})$ instances, and each instance has $\Order(\sqrt{n})$ vertices and $\Order(n)$
   edges. To implement this efficiently, we can sort the adjacency lists of $G$ according to the color classes, so that
   we can quickly extract the edges between any triplet of color classes. Then, it takes $\Order(n^{2.5})$ time to
   generate all the instances (which is linear in the output size).

   For every triangle $xyz$ in $G$, its vertices $x,y,z$ must belong to different color classes (two vertices cannot be
   in the same color class because an edge connects them). Since the color classes are disjoint, this specific triplet
   of color classes is unique. Thus, the triangle $xyz$ appears in exactly one output instance.
\end{proof}

The second tool is an isolation lemma that reduces the problem to instances with a unique solution by applying a
``random sieve'' to a colored instance. This technique dates back to the Valiant-Vazirani theorem~\cite{valiant1985np}
and is commonly used in various settings (e.g., the approximate counting
setting~\cite{goldreich2025coarse,dell2022approximately}). We use the following formulation
from~\cite{abboud2026triangle}:
\begin{lemma}[{\cite[Lemma 5.2]{abboud2026triangle}}]
   \label{lem:sieving}
   There is a near-linear time reduction that, given a $3$-colored graph $G$ with at most $n$ vertices in each color
   class, outputs $r = \tOrder(1)$ induced subgraphs of $G$, such that if $G$ contains a triangle, then with probability
   at least $1-\Order(1/n)$, one of the subgraphs contains a unique triangle.
\end{lemma}
The high-level idea is to guess (via exponential search) the approximate number of vertices participating in a triangle
within some color class, and then appropriately subsample. For example, if there are $t$ red vertices that participate
in a triangle, subsampling the red vertices with probability $\Theta(1/t)$ results in a graph that, with constant
probability, has only one red vertex participating in a triangle. Repeating this for every color class gives the
result.\footnote{We remark here that~\cite{abboud2026triangle} replaced two subsampling phases with one phase where they
subsample the edges. This is not suitable for induced $H$-free graphs, but as we explained, it can be replaced with
subsampling vertices at the cost of another $\Order( \log n )$ factor in the running time. In fact, this is the standard
way to prove such a result.} It is crucial here that the input graph is properly $3$-colored. If instead we merely
partitioned $V(G)$ into three parts that need not be independent (e.g., by a random partition), then the isolation lemma
would only ensure a unique \emph{colorful} triangle in one of the instances. Other triangles using a monochromatic edge
could still remain, so the resulting instance would not necessarily be a Unique Triangle Detection instance.

\begin{proof}[Proof of \cref{lem:reduction-to-unique}]
   Given the graph $G$, we first apply the color-coding reduction from \cref{lem:color-coding} to generate $\ell =
   \Order(n^{3/2})$ induced subgraphs $G_1, \ldots, G_{\ell}$, where each $G_i$ is $3$-colored on $\Order(\sqrt{n})$
   vertices. For each $G_i$, we apply the reduction from~\cref{lem:sieving} to generate $r = \tOrder(1)$ induced
   subgraphs $G_{i,1}, \ldots, G_{i,r}$. The collection of all $G_{i,j}$ constitutes the output.

   If $G$ is triangle-free, then all subgraphs are triangle-free. If $G$ has a triangle, \cref{lem:color-coding} ensures
   it appears in some $G_i$. Conditioned on this, \cref{lem:sieving} ensures that with high probability, one of the
   $G_{i,j}$ contains a unique triangle. The generation of the instances takes $\tOrder(n^{2.5})$ time, hence the
   overall time bound of the reduction is $\tOrder(n^{2.5})$.
\end{proof}

We are now ready to prove the main theorem in this section.

\begin{proof}[Proof of \cref{thm:induced-to-subgraph}]
   Assume there exists a Triangle Detection algorithm $\mathcal{A}$ for $\Hplus$-free graphs running in time at most
   $T(N)$ on $N$-vertex graphs and succeeds with probability at least $2/3$. Given an induced $H$-free graph $G$, we
   apply the reduction from \cref{lem:reduction-to-unique} to generate instances $G_1, \ldots, G_k$ where $k =
   \tOrder(n^{3/2})$ and each $G_i$ is an induced subgraph of $G$. The promise is that if $G$ contains a triangle, then
   with probability $1 - \Order(1/n)$, at least one $G_i$ contains a unique triangle. We run $\mathcal{A}$ on every
   instance $G_i$, and we amplify the success probability of $\mathcal{A}$ to be $1-\Order(1/n^3)$. Standard
   amplification is done by running $\Theta( \log n )$ independent executions of $\mathcal{A}$ on $G_i$ and taking the
   majority vote. Since there are only $\tOrder(n^{3/2})$ oracle calls, the total failure probability of all oracle
   calls is $\tOrder(n^{3/2}/n^3)=\Order(1/n)$. If the majority of executions return YES, we return YES for that
   instance; otherwise, we return NO.

   \paragraph*{Correctness.}
   We analyze the two cases depending on whether $G$ contains a triangle.
   \begin{itemize}
      \item \textbf{Case 1: $G$ is triangle-free (NO-instance).}
         Since every $G_i$ is an induced subgraph of $G$, it is also induced $H$-free and triangle-free. For our
         algorithm to err, $\mathcal{A}$ must output YES on some $G_i$. This can happen only if:
         \begin{enumerate}
         \item The randomized algorithm $\mathcal{A}$ fails. This happens with probability $\Order(1/n^3)$, for a total
            error probability of $\tOrder(n^{3/2}/n^3) = \Order(1/n)$.
         \item The input $G_i$ is not $\Hplus$-free, so the behavior of $\mathcal{A}$ on $G_i$ is indeterminate.
      \end{enumerate}
         We show that item 2 is impossible. Suppose, for contradiction, that $G_i$ contains a copy of $\Hplus$. Since
         $G_i$ is induced $H$-free, and $\Hplus$ contains $H$ as a subgraph, this (non-induced) copy of $H$ must induce
         an edge $uv$ such that $uv \notin E(H)$. By the definition of the augmented pattern $\Hplus$, the vertices $u,
         v$ are connected to a specific vertex $x_{uv}$ in $\Hplus$. Consequently, the presence of edge $uv$ implies
         that $\{u, v, x_{uv}\}$ forms a triangle in $G_i$. This contradicts the fact that $G_i$ is triangle-free. Thus,
         all inputs are valid $\Hplus$-free graphs, and the algorithm returns NO with probability $1-\Order(1/n)$.

      \item \textbf{Case 2: $G$ contains a triangle (YES-instance).}
         By \cref{lem:reduction-to-unique}, with probability $1-\Order(1/n)$, there exists an index $i$ such that $G_i$
         contains a unique triangle. We claim that such a graph $G_i$ is $\Hplus$-free. Suppose, for contradiction, that
         $G_i$ contains a copy of $\Hplus$. As in Case 1, since $G_i$ is induced $H$-free, this copy must use a non-edge
         $uv$ of $H$ as an edge in $G_i$. By the construction of $\Hplus$, the endpoints $u, v$ are connected to
         \emph{two} distinct vertices $x_{uv}$ and $y_{uv}$ in $\Hplus$. Therefore, the edge $uv$ closes two distinct
         triangles: $uvx_{uv}$ and $uvy_{uv}$. This contradicts the property that $G_i$ contains exactly one triangle.
         Thus, the instance $G_i$ is valid ($\Hplus$-free) and contains a triangle. The algorithm $\mathcal{A}$ will
         return YES with probability $1-\Order(1/n^3)$.
   \end{itemize}

   \paragraph*{Complexity analysis.}
   The total time is the reduction time plus the oracle calls: \[ \underbrace{\tOrder(n^{2.5})}_{\text{Reduction}} +
   \underbrace{\tOrder(n^{3/2})}_{\text{\# instances}} \times \underbrace{T(\Order(\sqrt{n}))}_{\text{Oracle time}}. \]
\end{proof}

%% file: sections/colored.tex
In this section, we prove the following theorem:
\begin{restatable}{theorem}{coloredtosubgraph}
   \label{thm:colored-to-subgraph}
   For every $3$-colored pattern $H$ that contains at most one triangle, there exists a $3$-colorable pattern $\Hstar$
   containing the same number of triangles as $H$, such that every colored $H$-free graph $G$ (with respect to the
   coloring of $H$) is also $\Hstar$-free.
\end{restatable}

We utilize an \emph{equality gadget} $EQ_{u,v}$ and an \emph{inequality gadget} $NEQ_{u,v}$, defined as follows. The
equality gadget is a $3$-colorable, triangle-free graph that, given two vertices $u,v$, forces $c(u) = c(v)$ in every
$3$-coloring $c$. We also generalize the gadget to sets $S$: $EQ_{u,S}$ forces $c(u) = c(v)$ for every $v \in S$, where
$S$ is an independent set. The requirements for the inequality gadget $NEQ_{u,v}$ are identical, except that it forces
$c(u) \neq c(v)$. The gadgets are derived from the Gr\"otzsch graph.

\begin{definition}
   The Gr\"otzsch graph $\mathcal{G}$ (named after Herbert Gr\"otzsch) is the smallest $4$-chromatic and triangle-free
   graph. It is also a $4$-critical graph, meaning that it becomes $3$-colorable upon the removal of any single edge.
   See \cref{fig:Grotzsch}.
\end{definition}

\begin{figure}[htbp]
   \centering
   \includegraphics[scale=1.2]{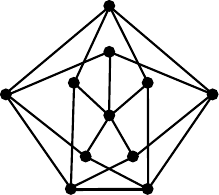}
   \caption{The Gr\"otzsch graph.}
   \label{fig:Grotzsch}
\end{figure}

\paragraph*{The equality gadgets $EQ_{u,v}$ and $EQ_{u,S}$.}
We construct $EQ_{u,v}$ by removing an edge $pq \in E(\mathcal{G})$ from the Gr\"otzsch graph and identifying $u$ and
$v$ with $p$ and $q$, respectively. Observe that $EQ_{u,v}$ is $3$-colorable and triangle-free. Moreover, the terminals
$p$ and $q$ must receive the same color in every $3$-coloring; otherwise, we could restore the edge $pq$ to obtain a
$3$-coloring of $\mathcal{G}$, violating the $4$-chromatic property of $\mathcal{G}$. $EQ_{u,S}$ is obtained by taking
$|S|$ independent copies of the basic gadget: $EQ_{u,v}$ for every $v \in S$.\footnote{This gadget can be made more
efficient in terms of size, but for the sake of brevity, we do not optimize it here.}

\paragraph*{The inequality gadget $NEQ_{u,v}$.}
The construction of $NEQ_{u,v}$ is similar, but we introduce a pendant vertex $w$ adjacent to $q$ (i.e., $w$ is adjacent
only to $q$). We identify $u$ and $v$ with $w$ and $p$ (instead of $q$ and $p$). This graph is $3$-colorable and
triangle-free (as adding pendant vertices does not affect colorability or triangle count). It forces $c(u) \neq c(v)$
because the edge $qw$ forces $c(q) \neq c(w)$, and since $c(q) = c(p)$ (from the underlying equality structure), it
follows that $c(w) \neq c(p)$.

Another gadget that we use in the proof is the following.

\begin{lemma}
   \label{lem:gadget}
   There exists a graph $X$ that is $3$-colorable and triangle-free, contains $3$ pairs of vertices $\set{x_1, x_2}$,
   $\set{y_1, y_2}$, and $\set{z_1, z_2}$, and satisfies the following properties:
   \begin{enumerate}
      \item The set $S = \set{x_1, x_2, y_1, y_2, z_1, z_2}$ is an independent set.
      \item In every $3$-coloring of $X$, each pair is assigned a distinct unordered pair of distinct colors.
   \end{enumerate}
\end{lemma}

For example, a $3$-coloring of $X$ might assign colors such that $\set{x_1, x_2}$ are red and blue, $\set{y_1, y_2}$ are
red and green, and $\set{z_1, z_2}$ are blue and green.

\begin{proof}[Proof of \cref{lem:gadget}]
   We construct $X$ as follows:

   \paragraph*{The core.}
   We create a ``core'' consisting of three independent vertices $u, v, w$ such that in every $3$-coloring, these
   vertices are assigned three distinct colors. We connect them pairwise using inequality gadgets: $NEQ_{u, v}, NEQ_{v,
   w}, \text{ and } NEQ_{w, u}$. These gadgets force $c(u) \neq c(v)$, $c(v) \neq c(w)$, and $c(w) \neq c(u)$. Since
   there are exactly three colors available, the set $\set{u, v, w}$ must be colored using all three colors in
   $\Colors$. Note that $\set{u, v, w}$ remains an independent set.

   \paragraph*{Connecting the pairs.}
   Let $S = \set{x_1, x_2, y_1, y_2, z_1, z_2}$ be a set of $6$ independent vertices. We connect elements of $S$ to the
   core vertices $u, v, w$ using standard edges. Since $c(u), c(v), c(w)$ are distinct, if a vertex $z \in S$ is
   connected to two vertices in $\set{u, v, w}$ (e.g., to $u$ and $v$), it must receive the color of the third (e.g.,
   $c(w)$). We add the following edges:

   \begin{description}
      \item[Pair 1:] Target colors for $(x_1, x_2)$ are $(c(u), c(v))$.
         \begin{itemize}
         \item Connect $x_1$ to $v$ and $w \implies c(x_1) = c(u)$.
         \item Connect $x_2$ to $u$ and $w \implies c(x_2) = c(v)$.
      \end{itemize}

      \item[Pair 2:] Target colors for $(y_1, y_2)$ are $(c(v), c(w))$.
         \begin{itemize}
         \item Connect $y_1$ to $u$ and $w \implies c(y_1) = c(v)$.
         \item Connect $y_2$ to $u$ and $v \implies c(y_2) = c(w)$.
      \end{itemize}

      \item[Pair 3:] Target colors for $(z_1, z_2)$ are $(c(u), c(w))$.
         \begin{itemize}
         \item Connect $z_1$ to $w$ and $v \implies c(z_1) = c(u)$.
         \item Connect $z_2$ to $u$ and $v \implies c(z_2) = c(w)$.
      \end{itemize}
   \end{description}

   \begin{figure}[htbp]
      \centering
      \includegraphics[scale=1]{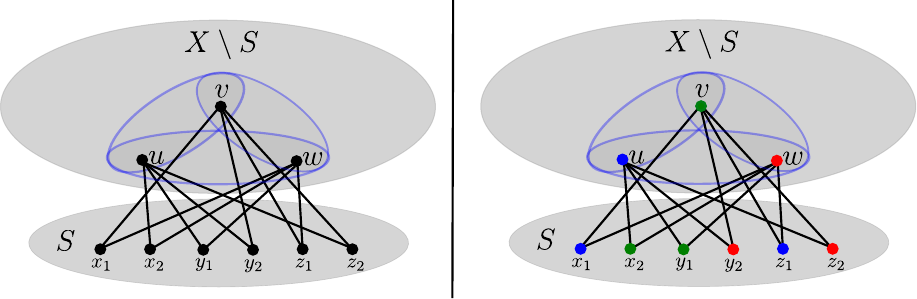}
      \caption{Left: The gadget graph $X$.
      The upper gray ellipse represents the core of the gadget, $X \setminus S$, containing the three uniquely colored
      vertices $u, v, w$ and the three inequality gadgets (blue ellipses). The lower gray ellipse represents the set $S$
      containing the three pairs $\set{x_1, x_2}, \set{y_1, y_2}, \set{z_1, z_2}$. \\ Right: The graph $X$ under a
      specific $3$-coloring $c$. Vertices $u, v, w$ receive distinct colors, which forces a unique coloring on $S$ such
      that each pair in $S$ is assigned a distinct pair of colors.}
      \label{fig:gadget_X}
   \end{figure}

   \paragraph*{Correctness.}
   \begin{itemize}
      \item \textbf{Triangle-Free:} The core uses triangle-free $NEQ$ gadgets.
         The new edges connect $S$ to $\set{u, v, w}$. Since both $S$ and $\set{u, v, w}$ are independent sets, no
         triangles are formed.
      \item \textbf{Distinct Pairs:} The pairs receive color sets $\set{c(u), c(v)}$, $\set{c(v), c(w)}$, and
         $\set{c(u), c(w)}$. These correspond to the three possible distinct unordered pairs of colors.
   \end{itemize}
\end{proof}

Let us now present the proof of the main theorem in this section.

\begin{proof}[Proof of \cref{thm:colored-to-subgraph}]
   Our construction of $\Hstar$ differs depending on whether $H$ contains a triangle. We begin with the case where $H$
   is triangle-free.

   \paragraph*{Construction for triangle-free $H$.}
   Denote by $c_H$ the coloring of $H$ and let the color classes of $H$ be $A, B, C$. The pattern $\Hstar$ is
   constructed by taking the gadget $X$ and attaching $6$ disjoint copies of $H$, one for each possible bijection
   between the pairs in $X$ to the color classes of $H$. For instance, consider a bijection $\phi$ defined by:
   $\set{x_1, x_2} \to B$, $\set{y_1, y_2} \to A$, and $\set{z_1, z_2} \to C$. For this assignment, we introduce an
   independent copy of $H$ denoted by $H_{\phi}$ and add all possible edges between the vertices of the pairs in $X$ and
   their assigned color classes in $H_{\phi}$. Specifically, we add all edges between $\set{x_1, x_2}$ and $B_{\phi}$,
   between $\set{y_1, y_2}$ and $A_{\phi}$, and between $\set{z_1, z_2}$ and $C_{\phi}$, where $A_{\phi}, B_{\phi},
   C_{\phi}$ are the copies of $A, B, C$ in $H_{\phi}$, respectively. We do so for every bijection $\phi$ ($6$
   bijections overall). See \cref{fig:h_star_full_example_pattern} for an example.

   First, observe that $\Hstar$ is triangle-free. The gadget $X$ is triangle-free by \cref{lem:gadget}, and every copy
   $H_{\phi}$ is triangle-free by assumption. Moreover, the only new edges are between a pair in the independent set $S
   \subseteq V(X)$ and a single color class of some copy $H_{\phi}$. Thus, any triangle using a new edge would have to
   contain either two vertices from $S$ or two vertices from the same color class of $H_{\phi}$; both are impossible,
   since $S$ and each color class of $H$ are independent.

   We now prove that if $G$ is colored $H$-free, then it is $\Hstar$-free. Suppose for contradiction that $G$ contains
   $\Hstar$ as a subgraph. The coloring of $G$ induces a $3$-coloring on the copy of $\Hstar$, and consequently on the
   copy of $X$ embedded within it. By \cref{lem:gadget}, each pair of vertices in $X$ is assigned a different pair of
   distinct colors.

   \begin{figure}[htbp]
      \centering
      \includegraphics[scale=1]{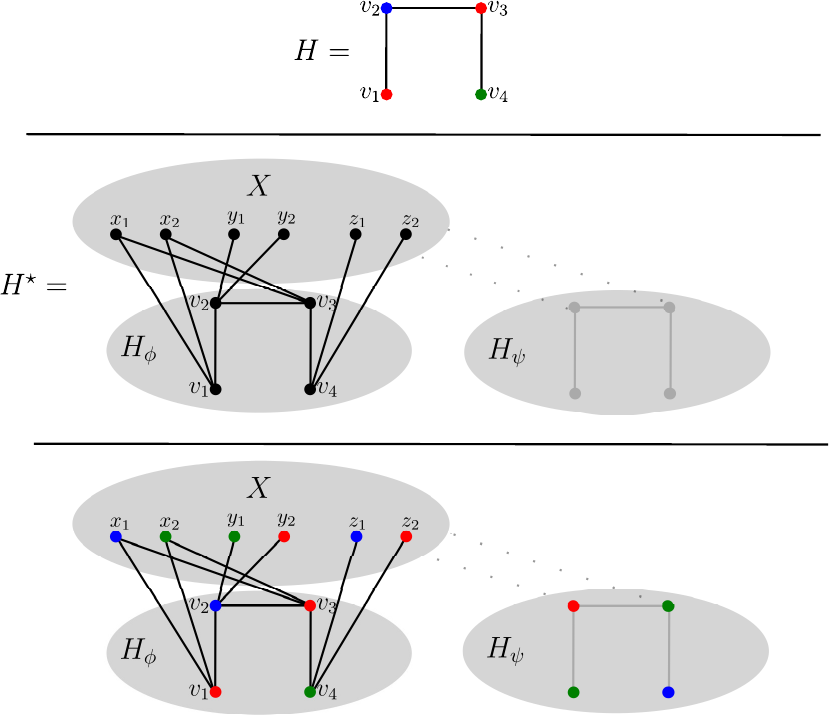}
      \caption{Top: The colored pattern $H$ given as input.
      Center: The structure of $\Hstar$, focusing on the connection between the gadget $X$ and a specific copy $H_\phi$.
      The connections enforce a specific coloring on $H_\phi$ based on the coloring of $X$. Bottom: The pattern $\Hstar$
      under a specific coloring $c$.}
      \label{fig:h_star_full_example_pattern}
   \end{figure}

   Without loss of generality, assume the coloring of $H$ assigns $\CRED$ to $A$, $\CBLUE$ to $B$, and $\CGREEN$ to $C$.
   Furthermore, assume without loss of generality that in the induced coloring of $X \subseteq G$:
   \begin{itemize}
      \item $\set{x_1, x_2}$ uses colors $\{\text{\CRED, \CBLUE}\}$,
      \item $\set{y_1, y_2}$ uses colors $\{\text{\CRED, \CGREEN}\}$,
      \item $\set{z_1, z_2}$ uses colors $\{\text{\CBLUE, \CGREEN}\}$.
   \end{itemize}

   Consider the copy of $H$ in $\Hstar$ corresponding to the assignment $\phi$ mapping $\set{z_1, z_2} \to A$,
   $\set{y_1, y_2} \to B$, and $\set{x_1, x_2} \to C$. Since $z_1$ and $z_2$ are colored with \CBLUE and \CGREEN, and
   they are fully connected to the vertices of $A_{\phi}$, the vertices of $A_{\phi}$ must be colored $\CRED$.
   Similarly, the vertices of $B_{\phi}$ must be $\CBLUE$ (as they are connected to the $\CRED$/$\CGREEN$ pair
   $\set{y_1, y_2}$), and the vertices of $C_{\phi}$ must be $\CGREEN$ (as they are connected to the \CRED/\CBLUE pair
   $\set{x_1, x_2}$). This implies the existence of a copy of $H$ in $G$ where $A$ is \CRED, $B$ is \CBLUE, and $C$ is
   \CGREEN. This is exactly the forbidden colored copy of $H$, contradicting the assumption that $G$ is colored
   $H$-free.

   \paragraph*{Construction for \texorpdfstring{$H$}{H} that contains a triangle.}
   The previous construction fails if $H$ contains a triangle, because creating $\Hstar$ by simply attaching copies of
   $H$ would result in $\Hstar$ containing multiple triangles. To avoid this, the new pattern $\Hstar$ is constructed so
   that all embedded copies of $H$ share a single, common triangle.

   Let $x, y, z$ be the vertices of the unique triangle in $H$, and let $H' = H - \set{x, y, z}$, the graph obtained by
   removing the triangle vertices from $H$. Let $c_H$ be the fixed coloring of $H$. The graph $\Hstar$ is constructed as
   follows:
   \begin{enumerate}
      \item \textbf{The core triangle:} Start with a central triangle consisting of three vertices $u, v, w$.
      \item \textbf{The $H$-candidates:} Add $6$ disjoint copies of $H'$. Each copy corresponds to a bijection $\phi
         \colon \set{u, v, w} \to \set{x, y, z}$. We interpret $\phi(u)$ as the role assigned to $u$ (e.g., if
         $\phi(u)=x$, then $u$ acts as $x$). Let $H'_{\phi}$ denote the copy associated with $\phi$.
      \item \textbf{Core-$H'_{\phi}$ connections:} For each bijection $\phi$, we connect the core $u, v, w$ to
         $H'_{\phi}$ as follows:
         \begin{itemize}
         \item \emph{Equality gadgets:} Let $S_{\phi(u)}$ be the set of vertices in $H'_{\phi}$ that corresponds to
            the color class of $\phi(u)$ in $H'$ (i.e., vertices $p \in H'$ with $c_H(p) = c_H(\phi(u))$). We add an
            equality gadget $EQ_{u, S_{\phi(u)}}$. We do the same for $v$ and $w$: add $EQ_{v, S_{\phi(v)}}$ and $EQ_{w,
            S_{\phi(w)}}$.
         \item \emph{Core-$H'_{\phi}$ edges:} We add edges between the core vertices and $H'_{\phi}$ to replicate the
            structure of $H$. Specifically, for every vertex $p'$ in $H'_{\phi}$ corresponding to $p \in H'$:
            \begin{itemize}
            \item If $p$ is a neighbor of $\phi(u)$ in $H$, add the edge $p'u$.
            \item If $p$ is a neighbor of $\phi(v)$ in $H$, add the edge $p'v$.
            \item If $p$ is a neighbor of $\phi(w)$ in $H$, add the edge $p'w$.
         \end{itemize}
      \end{itemize}
   \end{enumerate}
   Crucially, observe that in the coloring $c_H$ of $H$, a vertex is never adjacent to vertices in its own color class.
   Therefore, there are no edges between $u$ and $S_{\phi(u)}$ (and similarly for $v, w$).

   \paragraph*{Correctness.}
   Let us start by proving that $\Hstar$ is valid.
   \begin{itemize}
      \item \textbf{Unique triangle:} $\Hstar$ contains exactly one triangle, $uvw$.
         To see this, consider any $\phi$. The subgraph induced by $\set{u, v, w} \cup V(H'_{\phi})$, ignoring the
         equality gadgets, is isomorphic to $H$ (which has a unique triangle $xyz$) by construction. Now, consider the
         equality gadgets. Each gadget is triangle-free. Furthermore, we never add edges between the endpoints of a
         gadget; for example, there are no edges between $u$ and $S_{\phi(u)}$. Therefore, the gadgets do not close any
         new triangles with the rest of the graph.

      \item \textbf{3-colorability:} The core triangle $uvw$ is $3$-colorable. Fix any coloring of $u, v, w$.
         For every $\phi$, we can extend this coloring to $H'_{\phi}$ as follows: assign all vertices in $S_{\phi(u)}$
         the color of $u$, all vertices in $S_{\phi(v)}$ the color of $v$, and all vertices in $S_{\phi(w)}$ the color
         of $w$. This is a valid coloring for $H'_{\phi}$ because the sets $S_{\phi(\cdot)}$ are independent sets in
         $H'$ (they are color classes), and there are no edges between them that violate the coloring (edges in $H'$
         only connect different color classes). Finally, we must verify the connections to the core. Edges connect $u$
         only to neighbors of $\phi(u)$, which have different colors from $\phi(u)$ (and thus different from $u$). The
         equality gadgets connect $u$ to $S_{\phi(u)}$; since both endpoints receive the same color, and the equality
         gadget admits a $3$-coloring where terminals have the same color, the gadgets can be properly colored.
   \end{itemize}

   We now prove that if $G$ contains $\Hstar$ as a subgraph, then $G$ is not colored $H$-free. Suppose $G$ contains a
   copy of $\Hstar$. The coloring $c_G$ of $G$ induces a coloring on the core triangle $uvw$. Since they form a
   triangle, $u, v, w$ receive distinct colors.

   \paragraph*{Example.} Suppose that in $H$, the vertices $x, y, z$ are colored $c_H(x) = \CRED$, $c_H(y) = \CBLUE$,
   and $c_H(z) = \CGREEN$, respectively. Suppose that in $G$, the core vertices are colored $c_G(u)=\CRED$,
   $c_G(v)=\CBLUE$, $c_G(w)=\CGREEN$. We select the copy of $H'$ associated with the mapping $\phi(u)=x$, $\phi(v)=y$,
   $\phi(w)=z$. The equality gadgets force $S_x$ (the red class of $H'$) to be colored the same as $u$ (\CRED), $S_y$ to
   be colored the same as $v$ (\CBLUE), and $S_z$ to be colored the same as $w$ (\CGREEN). The added edges connect $u$
   (\CRED) to the neighbors of $x$ (which are non-red). Thus, the subgraph $uvw \cup V(H'_{\phi})$ forms a copy of $H$
   where vertices have exactly the colors prescribed by $c_H$.

   \paragraph*{General argument.} Let $\phi$ be the mapping defined by the imposed coloring of $G$, mapping a vertex
   $\ell \in \set{u, v, w}$ to the unique vertex in $\set{x, y, z}$ such that $c_G(\ell) = c_H(\phi(\ell))$. Consider
   the subgraph $\set{u, v, w} \cup V(H'_{\phi})$. The equality gadgets force the color classes of $H'_{\phi}$ to match
   the colors of the corresponding core vertices. Specifically, vertices in $S_{\phi(u)}$ must have color $c_G(u)$,
   which by our choice of $\phi$ is $c_H(\phi(u))$. Consequently, the coloring of this subgraph in $G$ is isomorphic to
   the coloring $c_H$ of $H$. Thus, $G$ contains a colored copy of $H$.
\end{proof}

%% file: sections/equiv.tex
In this section, we prove the main theorem of this paper, \cref{thm:main}, showing that Hypotheses
\labelcref{hypothesis:subgraph}, \labelcref{hypothesis:induced,hypothesis:colored} are all equivalent.

\begin{proof}[Proof of \cref{thm:main}]
   We begin by showing an equivalence between \cref{hypothesis:subgraph} (the subgraph setting) and
   \cref{hypothesis:induced} (the induced setting), using \cref{thm:induced-to-subgraph}.

   \paragraph*{Subgraph $\implies$ Induced.}
   Suppose that \cref{hypothesis:subgraph} is true. Recall that the lower bounds from~\cite{abboud2026triangle}
   naturally extend to the induced setting because the class of $H$-free graphs is a subset of the class of induced
   $H$-free graphs. Thus, if the problem is ``hard'' for $H$-free graphs, it is certainly ``hard'' for the superset.

   It remains to show that the ``easy'' patterns also extend to the induced setting. Let $H$ be a $3$-colorable pattern
   containing at most one triangle. Since $\Hplus$ preserves $3$-colorability and the triangle count, our assumption
   that~\cref{hypothesis:subgraph} holds implies that Triangle Detection in $\Hplus$-free graphs is solvable in
   $n^{3-\eps}$ time, for some $\eps > 0$. Consequently, by \cref{thm:induced-to-subgraph}, there exists an algorithm
   for induced $H$-free graphs running in truly subcubic time $\tOrder(n^{2.5} + n^{3/2} (\sqrt{n})^{3-\eps}) =
   \tOrder(n^{2.5}+n^{3-\eps/2})$. Therefore, if \cref{hypothesis:subgraph} holds, then \cref{hypothesis:induced} holds
   as well.

   \paragraph*{Induced $\implies$ Subgraph.}
   Conversely, a confirmation of \cref{hypothesis:induced} naturally implies a confirmation of
   \cref{hypothesis:subgraph}.

   The ``hard'' patterns are already established in~\cite{abboud2026triangle}. For the ``easy'' patterns, any algorithm
   that solves Triangle Detection in induced $H$-free graphs also works for $H$-free graphs, as the latter is a subset
   of the former. Thus, the two hypotheses are equivalent.

   \medskip
   We next establish the equivalence between \cref{hypothesis:subgraph} (the subgraph setting) and
   \cref{hypothesis:colored} (the colored setting), using \cref{thm:colored-to-subgraph}.

   \paragraph*{Subgraph $\implies$ Colored.}
   Suppose that \cref{hypothesis:subgraph} is true. First, consider the ``hard'' patterns. For $3$-colorable patterns
   $H$ that contain more than one triangle, the lower bounds from~\cite{abboud2026triangle} extend to the colored
   setting via color-coding. Specifically, take any $H$-free graph $G$. We can color-code it (e.g., using random
   partitioning or a perfect hash family~\cite{alon1995color}) to produce a $3$-colored subgraph of $G$ that preserves
   triangles with constant probability. Moreover, such subgraph is colored $H$-free for every possible coloring of $H$,
   since $G$ contained no copy of $H$ at all. Thus, a truly subcubic algorithm for colored $H$-free graphs would imply,
   by standard repetition of the color-coding step, a truly subcubic algorithm for $H$-free graphs that succeeds with
   high probability. This implies that patterns containing more than one triangle are ``hard'' in the colored setting as
   well.

   It remains to show that the ``easy'' patterns also extend. Let $H$ be a $3$-colored pattern containing at most one
   triangle. By \cref{thm:colored-to-subgraph}, there exists a $3$-colorable pattern $\Hstar$ containing the same number
   of triangles as $H$ such that every colored $H$-free graph is also $\Hstar$-free. By our assumption that
   \cref{hypothesis:subgraph} holds, Triangle Detection in $\Hstar$-free graphs is solvable in $n^{3-\eps}$ time, for
   some $\eps > 0$. Hence, we can run the algorithm for $\Hstar$-free graphs on $G$ to solve Triangle Detection in
   colored $H$-free graphs in $n^{3-\eps}$ time. Consequently, if \cref{hypothesis:subgraph} holds,
   \cref{hypothesis:colored} holds as well.

   \paragraph*{Colored $\implies$ Subgraph.}
   Conversely, suppose that \cref{hypothesis:colored} is true. The ``hard'' patterns are already established
   in~\cite{abboud2026triangle}. For the ``easy'' patterns, let $H$ be a $3$-colorable pattern with at most one
   triangle. Given an uncolored $H$-free graph $G$, we color-code it (as in the previous paragraph) to obtain a
   $3$-colored subgraph $G'$ of $G$ that preserves triangles with high probability. Fix an arbitrary $3$-coloring of $H$
   (such a coloring exists since $H$ is $3$-colorable). Since $G$ is $H$-free, $G'$ contains no colored copies of $H$.
   By the assumption that~\cref{hypothesis:colored} holds, there exists a truly subcubic algorithm for colored $H$-free
   graphs. Hence, by standard repetition of the color-coding step, there is a truly subcubic algorithm that decides
   whether $G$ contains a triangle with high probability. Thus, $H$ is an ``easy'' pattern in the subgraph setting, and
   the two hypotheses are equivalent.

   This completes the proof, showing that Hypotheses \labelcref{hypothesis:subgraph},
   \labelcref{hypothesis:induced,hypothesis:colored} are all equivalent.
\end{proof}

%% file: sections/applications.tex
Notably, our reduction for induced subgraphs (\cref{thm:induced-to-subgraph}) serves not only as a bridge between
\cref{hypothesis:subgraph} and \labelcref{hypothesis:induced}, but also as an algorithmic tool for induced patterns that
do not a priori seem to admit a truly subcubic algorithm. We observe that while it is unclear how to efficiently detect
a triangle in an induced $P_6$-free graph or an induced $C_7$-free graph, their augmented patterns admit a
\emph{degenerate coloring}.

\begin{definition}
   \label{def:degenerate}
   A $3$-colored graph $H$ containing at most one triangle is \emph{degenerately colored} if, in every induced subgraph
   of $H$ that is not a triangle, there exists a vertex with a monochromatic neighborhood.
\end{definition}

It was shown in~\cite[Theorem 2]{abboud2026triangle} that if a triangle-free pattern $H$ admits a degenerate coloring,
then Triangle Detection in $H$-free graphs is solvable in $\tOrder(n^{3-2^{-(k-3)}})$ time, where $k = |V(H)|$. By
applying the reduction from \cref{thm:induced-to-subgraph}, we obtain the following result:

\begin{proposition}
   \label{prop:P6-C7}
   Let $a = |V(P_6^+)| = 26$ and $b = |V(C_7^+)| = 35$. There is a combinatorial algorithm for Triangle Detection in
   induced $P_6$-free graphs and induced $C_7$-free graphs, running in time \[ \tOrder\left(n^{3/2} \cdot
   \sqrt{n^{3-2^{-(a-3)}}}\right) = \tOrder(n^{3-2^{-(a-2)}}), \] and \[ \tOrder\left(n^{3/2} \cdot
   \sqrt{n^{3-2^{-(b-3)}}}\right)= \tOrder(n^{3-2^{-(b-2)}}), \] respectively.
\end{proposition}

The degenerate colorings of $P_6^{+}$ and $C_7^{+}$ are illustrated in \cref{fig:P6,fig:C7}, respectively.

\begin{figure}[htbp]
   \centering
   \includegraphics[scale=1.0]{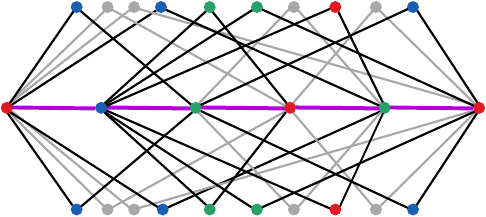}
   \caption{The augmented pattern $P_6^+$ with a degenerate coloring.
   The edges of the original $P_6$ are represented by highlighted purple lines. Wedges connecting the endpoints of
   non-edges are shown in gray whenever the endpoints share the same color; this is because such paths do not obstruct
   the degenerate coloring. From this coloring, it is evident that $P_6^+$ admits a degenerate coloring. Consequently,
   Triangle Detection in induced $P_6$-free graphs can be solved in truly subcubic time using~\cite[Theorem
   2]{abboud2026triangle}.}
   \label{fig:P6}
\end{figure}

\begin{figure}[htbp]
   \centering
   \includegraphics[scale=0.8]{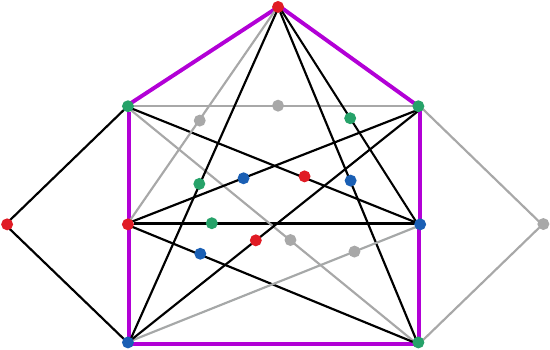}
   \caption{The augmented pattern $C_7^+$ with a degenerate coloring.
   The edges of the original $C_7$ are represented by highlighted purple lines. Paths connecting the endpoints of
   non-edges are shown in gray whenever the endpoints share the same color; this is because such paths do not obstruct
   the degenerate coloring. (Note: For visual clarity, only one wedge is shown for these connections, rather than the
   two required by the formal definition of $\Hplus$). From this coloring, it is evident that $C_7^+$ admits a
   degenerate coloring. Consequently, Triangle Detection in induced $C_7$-free graphs can be solved in truly subcubic
   time using~\cite[Theorem 2]{abboud2026triangle}.}
   \label{fig:C7}
\end{figure}

While our reduction to induced subgraphs yields some unexpected algorithmic results, the reduction to colored subgraphs
does not achieve similar results. The reason for this is that the augmented pattern $\Hstar$ contains a subgraph that
does not admit a degenerate coloring. Specifically, the equality gadget is derived from the Gr\"otzsch graph by removing
an edge, which is not degenerately colorable. Therefore, we cannot use~\cite[Theorem 2]{abboud2026triangle} in
conjunction with \cref{thm:colored-to-subgraph} to obtain new algorithms for colored patterns.

%% file: sections/k_clique.tex
In this section, we show how our results can be extended from Triangle Detection to $k$-Clique Detection, i.e., from
$k=3$ to any $k\geq 3$. The motivation for studying $k$-Clique in $H$-free graphs closely parallels the $k=3$ case
discussed in the introduction. The brute-force $\Order(n^3)$ algorithm for Triangle Detection extends to a brute-force
$\Order(n^k)$ algorithm for $k$-Clique. It has been a major open question to break this barrier without using FMM, and
the conjecture that this is impossible was put forth in 2015 by Abboud, Backurs, and Vassilevska Williams \cite{ABV15}
and subsequently used as the basis for many conditional lower bounds (e.g.,
\cite{bringmann2017dichotomy,chang2019hardness,lincoln2018tight,bringmann2017clique}).

\begin{conjecture}[Combinatorial $k$-Clique Detection]
   \label{conj:k-clique}
   For any $\eps > 0$, there is no combinatorial $k$-Clique Detection algorithm that runs in $\Order(n^{k-\eps})$ time.
\end{conjecture}

Similarly to the $k=3$ case, one hopes that breaking the $n^k$ barrier with a new algorithm that avoids FMM would lead
to a breakthrough for one of the generalizations of the problem, such as the weighted case where no known algorithm
(even with FMM) can find the minimum weight $k$-Clique in $\Order(n^{k-\eps})$ time. The latter was shown to be a
barrier to faster algorithms for several problems (e.g.,
\cite{abboud2014consequences,backurs2017improving,backurs2016tight,bringmann2020tree}). Understanding the complexity of
$k$-Clique in $H$-free graphs could also inform algorithmic attacks on $k$-Clique that employ the
structure-vs-randomness paradigm. We refer the reader to \cite{AFS24} for further discussion on the $k$-Clique
conjecture.

The natural extension of the dichotomy hypothesis of \cite{abboud2026triangle} from the $k=3$ case to $k$-Clique is the
following.

\begin{hypothesis}[Dichotomy for $k$-Clique Detection in $H$-Free Graphs]
   \label{hypothesis:K-Clique-H-free}
   The $k$-Clique Detection problem in $H$-free graphs is:
   \begin{itemize}
      \item solvable in $\Order(n^{k-\eps})$ time via a combinatorial algorithm for some $\eps>0$, if $H$ is
         $k$-colorable and contains at most one $k$-clique; and
      \item as hard as $k$-Clique Detection in general graphs, otherwise.
   \end{itemize}
\end{hypothesis}

By simple extensions of the results of \cite{abboud2026triangle}, one arrives at a similar situation for this hypothesis
compared to the $k=3$ case: the hardness part can be proven (i.e., patterns that are not $k$-colorable or contain more
than one $k$-clique are ``hard''), whereas the algorithmic part is only partially known (several positive results can be
obtained, but they do not cover all patterns).

The main result of this section is to show that the primary message of this paper applies to $k$-Clique as well: the
\emph{induced} variant of the dichotomy hypothesis holds if and only if the non-induced variant holds.

Formally, the main result in this section is the following theorem:
\begin{theorem}
   \label{thm:induced-to-subgraph-k-clique}
   For every $k$-colorable pattern $H$ containing at most one $k$-clique, there exists a $k$-colorable pattern
   $\Hplus_k$ containing the same number of $k$-cliques as $H$, such that $k$-Clique Detection in induced $H$-free
   graphs reduces to $k$-Clique Detection in $\Hplus_k$-free graphs. Specifically, suppose there is an algorithm for
   $k$-Clique Detection in $\Hplus_k$-free graphs that runs in time $T(N)$ on $N$-vertex graphs and succeeds with
   probability at least $2/3$. In that case, there is an algorithm for $k$-Clique Detection in induced $H$-free graphs
   running in time \[ \tOrder\left(n^{k - \frac{k-2}{k-1}} + n^{k(1-\frac{1}{k-1})} \cdot T(n^{\frac{1}{k-1}}) \right),
   \] and succeeding with probability at least $1-\Order(1/n)$.
\end{theorem}
Note that the theorem implies that if $T(N) \leq N^{k-\eps}$, then the running time we get is polynomially smaller than
$n^k$: $n^{k-\eps/(k-1)}$.

The generalization of the pattern $\Hplus$ to $\Hplus_k$ is as follows:

\begin{definition}
   Given a pattern $H$, the \emph{augmented pattern} $\Hplus_k$ is obtained by attaching two $(k-2)$-cliques to the
   endpoints of every non-edge of $H$. Formally, $V(H) \subseteq V(\Hplus_k)$, $E(H) \subseteq E(\Hplus_k)$, and for
   every $u,v \in V(H)$ such that $uv \notin E(H)$, we add two distinct $(k-2)$-cliques $X_{uv}, Y_{uv}$ and edges
   $xu,xv,yu,yv$ for every $x \in X_{uv}$ and $y \in Y_{uv}$.
\end{definition}

Note that $\Hplus_k$ admits the following structural properties: If $H$ is $k$-colorable, so is $\Hplus_k$, since
$X_{uv}$ and $Y_{uv}$ see at most two colors in a $k$-coloring of $H$ and thus can be colored with the remaining $k-2$
colors. The number of $k$-cliques in $\Hplus_k$ is preserved (i.e., it is the same as in $H$). This is because $X_{uv}$
and $Y_{uv}$ are not adjacent to any edge in $H$.

To prove the theorem, we will need a reduction to Unique $k$-Clique Detection that preserves induced $H$-freeness.

\begin{lemma}[Reduction to Unique $k$-Clique Detection]
   \label{lem:reduction-to-unique-k-clique}
   There is a randomized reduction that, given a graph $G$ with $n$ vertices and $m$ edges, outputs
   $q=\tOrder(n^{k(1-\frac{1}{k-1})})$ induced subgraphs $G_1, \ldots, G_q$, such that:
   \begin{enumerate}
      \item If $G$ contains a $k$-clique, then with probability $1 - \Order(1/n)$, there exists at least one $G_i$ that
         contains exactly one $k$-clique.
      \item Each $G_i$ has $\Order(n^{\frac{1}{k-1}})$ vertices.
   \end{enumerate}
   The reduction runs in $\tOrder(m n^{(k-2)(1- \frac{1}{k-1})})$ time.
\end{lemma}

As with triangles, the implementation is via color-coding that preserves induced $H$-freeness, followed by an isolation
step. We present a color-coding algorithm below which is essentially a generalization of~\cite{dumitrescu2025finding} to
$k$-cliques.

\begin{lemma}[$k$-Clique or $\Order(n^{1-\frac{1}{k-1}})$-Coloring]
   \label{lem:clique-or-coloring}
   For every $k \geq 2$, there is an algorithm that runs in time $\Order(n+m)$ and finds either a $k$-clique or an
   $\Order(n^{1-\frac{1}{k-1}})$-coloring.
\end{lemma}
\begin{proof}
   The algorithm is recursive. We prove, by induction on $k$, the stronger statement that if the algorithm does not find
   a $k$-clique, then it outputs a coloring with at most $k n^{1-\frac{1}{k-1}}$ colors.

   For $k = 2$ the claim is trivial: we just have to find a single edge in the graph. If no edge is found, it means the
   graph is an independent set of size $n$ and can be colored with one color. For $k > 2$, let $\Delta =
   n^{1-\frac{1}{k-1}}$. As long as there is some vertex $v$ whose degree is at least $\Delta$, choose an arbitrary
   subset $U \subseteq N(v)$ of size $\Delta$, and apply the recursive algorithm to $G[U]$ with $k' = k-1$. If a
   $(k-1)$-clique is found in $G[U]$, we are done because, together with $v$, it forms a $k$-clique. Otherwise, by the
   induction hypothesis, we obtain a coloring of $G[U]$ with at most \[ (k-1)|U|^{1-\frac{1}{k-2}} =
   (k-1)\Delta^{1-\frac{1}{k-2}} \] distinct colors. We remove $U$ from the graph and repeat this process, each time
   using a fresh set of colors. This phase lasts at most $n/\Delta$ iterations because in each iteration we remove
   $\Delta$ vertices from the graph. Hence, the number of colors produced in this phase is at most \[ \frac{n}{\Delta}
   \cdot (k-1)\Delta^{1-\frac{1}{k-2}} = (k-1)\frac{n}{\Delta^{\frac{1}{k-2}}} = (k-1)n^{1-\frac{1}{k-1}}. \] After the
   first phase is over, the maximum degree is less than $\Delta$. Then, the greedy coloring algorithm produces at most
   $\Delta = n^{1-\frac{1}{k-1}}$ distinct colors (again using a fresh set of colors) in linear time~\cite[p.
   147]{bollobas1998modern}. In total, the number of colors used in both phases is at most $k n^{1-\frac{1}{k-1}}$.

   To see that the running time of this algorithm is linear, consider the first phase of the algorithm (as we already
   noted, the second phase takes linear time). Observe that the sets $U_1,U_2,\ldots$ of size $\Delta$ chosen at one
   recursive level are vertex-disjoint. Hence, the total time needed to construct the instances $G[U_1],G[U_2],\ldots$
   (by scanning, for each $i$ and each $u \in U_i$, the adjacency list of $u$) is $\Order(n+m)$. By induction, the
   recursive calls on these instances also take linear time in their sizes, which, by disjointness, sum to
   $\Order(n+m)$. Since the number of recursive levels is at most $k$, and $k$ is fixed, we conclude that the total
   running time across all levels is $\Order(m+n)$.

\end{proof}

Given \cref{lem:clique-or-coloring}, color-coding that preserves $H$-freeness follows:

\begin{lemma}[Color-coding for $k$-Clique Detection]
   \label{lem:color-coding-k-clique}
   There is a deterministic algorithm that reduces $k$-Clique Detection to $k$-colored $k$-Clique Detection, running in
   $\Order(m n^{(k-2)(1- \frac{1}{k-1})})$ time. It either finds a $k$-clique in $G$, or outputs
   $\Order(n^{k(1-\frac{1}{k-1})})$ instances, where each instance is a $k$-colored induced subgraph of $G$ on
   $\Order(n^{\frac{1}{k-1}})$ vertices. Every $k$-clique in $G$ is present in exactly one instance.
\end{lemma}
\begin{proof}
   The proof is similar to the case of $k=3$: We compute a coloring using \cref{lem:clique-or-coloring}, refine the
   partition to reduce large color classes to $\Order(n^{\frac{1}{k-1}})$ vertices, and then produce instances for each
   $k$-tuple of colors. The number of output instances is $\Order(n^{k(1-\frac{1}{k-1})})$. The running time is governed
   by the output size, which is $\Order(m n^{(k-2)(1-\frac{1}{k-1})})$ because each edge participates in
   $\Order((n^{1-\frac{1}{k-1}})^{k-2})$ $k$-tuples.
\end{proof}

The next ingredient is an isolation lemma for $k$-Cliques:

\begin{lemma}[Isolation for $k$-colored graphs]
   \label{lem:sieving-k-color}
   There is a near-linear time reduction that, given a $k$-colored graph $G$ with at most $n$ vertices in each part,
   outputs $\Order(\log^k(n))$ induced subgraphs of $G$, such that if $G$ contains a $k$-clique, then with probability
   at least $1-\Order(1/n)$, one of the subgraphs contains a unique $k$-clique.
\end{lemma}

We sketched the proof of this lemma for $k=3$ in \cref{sec:induced}. For larger $k$, the generalization is
straightforward. We apply the ``random sieve'' again: Subsample each color class independently with probabilities taken
from $\set{2^{-i} \mid 0 \leq i \leq \lceil \log_2 n \rceil}$. It is then possible to show that with constant
probability, some subgraph obtained by the sampling procedure contains exactly one $k$-clique if $G$ contains a
$k$-clique. With another $\Order(\log n)$ repetitions, we can boost the success probability to $1-\Order(1/n)$.

\begin{proof}[Proof of \cref{lem:reduction-to-unique-k-clique}]
   By applying \cref{lem:color-coding-k-clique} and then \cref{lem:sieving-k-color}, we obtain a reduction to Unique
   $k$-Clique Detection that preserves the induced $H$-freeness property (as each output instance is an induced subgraph
   of $G$). The running time is that of color-coding times a $\tOrder(1)$-factor for applying the isolation lemma to
   every graph in the output. Thus, the running time is $\tOrder( m n^{(k-2)(1-\frac{1}{k-1})})$, proving
   \cref{lem:reduction-to-unique-k-clique}.
\end{proof}

\begin{proof}[Proof of \cref{thm:induced-to-subgraph-k-clique}]
   We apply the reduction to Unique $k$-Clique Detection from \cref{lem:reduction-to-unique-k-clique} and run the
   supposed algorithm on every output graph $G_i$, returning YES if at least one execution returns YES. If $G$ is
   $k$-clique-free, then every output graph is also $k$-clique-free. In this case each $G_i$ is $\Hplus_k$-free:
   otherwise, because $G_i$ is induced $H$-free, a copy of $\Hplus_k$ would use a non-edge of $H$ as an edge and would
   therefore close a $k$-clique with one of the attached $(k-2)$-cliques. If $G$ contains a $k$-clique, then with high
   probability some output graph $G_i$ contains a unique $k$-clique. This graph is $\Hplus_k$-free, since a copy of
   $\Hplus_k$ in an induced $H$-free graph would close two distinct $k$-cliques using the two attached $(k-2)$-cliques,
   contradicting uniqueness. The total running time of the reduction is \[ \Order( m n^{(k-2)(1-\frac{1}{k-1})} +
   n^{k(1-\frac{1}{k-1})} T(n^{\frac{1}{k-1}}) ). \] Plugging in $m = \Order(n^2)$ gives the result. This completes the
   proof of \cref{thm:induced-to-subgraph-k-clique} (more details appear in the proof for $k=3$ in \cref{sec:induced}).
\end{proof}

Finally, the formal equivalence between \cref{hypothesis:K-Clique-H-free} and its induced analog is established: The
``hard'' patterns in the subgraph setting extend to the induced setting naturally, while the ``easy'' patterns extend
via \cref{thm:induced-to-subgraph-k-clique}. The proof of the other direction (induced $\implies$ subgraph) is
straightforward. The details are similar to the $k=3$ case (see \cref{sec:equiv}).

%% file: sections/conclusions.tex
We have shown that the three hypotheses: \labelcref{hypothesis:subgraph},
\labelcref{hypothesis:induced,hypothesis:colored}, are all equivalent. Consequently, the central open question remains
the validity of these hypotheses. Our results suggest a strategic approach for resolving this: to prove the hypotheses,
it suffices to focus on the restricted subgraph setting, which is arguably simpler. Conversely, to refute them, one can
target the more general colored or induced settings. For instance, demonstrating a cubic conditional lower bound for
induced $P_{100}$-free graphs would refute the dichotomy hypothesis for the subgraph setting as well.

Even if the hypotheses hold true, we can ask a more refined question regarding exact running times. In other words,
while our work shows that the hypotheses are equivalent in classifying patterns as ``easy'' or ``hard'', it does not
provide a complete picture of the exact running times for specific patterns. Specifically, for an ``easy'' pattern $H$,
are induced $H$-free graphs strictly harder to process than their subgraph counterparts? Our reduction implies a
slowdown due to color-coding and the increased size of the augmented pattern $\Hplus$. For example, Triangle Detection
in (subgraph) $C_7$-free graphs is solvable in $\tOrder(m+n^{5/3})$ time~\cite[Theorem 6]{abboud2026triangle}, but the
running time we obtain for induced $C_7$-free graphs is significantly larger: $\tOrder(n^{3-2^{-33}})$.

Note that even if the hypotheses are false and a different dichotomy emerges, our reduction may still have the potential
to prove equivalence between the induced and non-induced settings. Because the augmented pattern $\Hplus$ is
structurally simple, it may still preserve the properties relevant to an alternative dichotomy.

Finally, another open question is whether we can obtain a deterministic reduction from the induced setting to the
non-induced setting. Our current reduction is based on the random sieve approach, and derandomizing this technique
remains an open problem.